\documentclass[
submission
, nomarks
]{dmtcs-episciences}

\usepackage[utf8]{inputenc}
\usepackage{subfigure}

\usepackage{amsmath}
\usepackage{tabularx}
\usepackage[algo2e]{algorithm2e}
\newtheorem{theorem}{Theorem}

\usepackage[round]{natbib}

\author{Yuriy Shablya\affiliationmark{1}
  \and Dmitry Kruchinin\affiliationmark{2}}
\title[Algorithms for ranking and unranking RNA secondary structures]{Algorithms for ranking and unranking \\ the combinatorial set of RNA secondary structures}%\thanks{The study of generating functions and explicit formulas was funded by RFBR (project no. 20-31-70037). The development of combinatorial generation algorithms was supported by the Russian Science Foundation (project no. 18-71-00059).}}
\affiliation{
  Tomsk State University of Control Systems and Radioelectronics, Tomsk, Russia}
\begin{document}
% This is only used if you are compiling for a volume before vol 25
% \publicationdetails{VOL}{2015}{ISS}{NUM}{SUBM}
% This is the new form of collecting the data, starting with vol 25
\publicationdata
{vol. XXX:XXX}
{2023}
{XXX}
{10.46298/dmtcs.XXX}
%{1998-10-14; 1998-10-14; 2002-07-19; 2014-02-05; 2015-09-09; 2022-12-25}
%{2022-12-3}
{2023-1-27}%; None}
{XXX}
\maketitle
\begin{abstract}
In this paper, we study the combinatorial set of RNA secondary structures of length $n$ with $m$ base-pairs.
For~a~compact representation, we encode an RNA secondary structure by the corresponding Motzkin word.
For this combinatorial set, we construct an~AND/OR tree structure, find a bijection between the combinatorial set and the set of variants of the AND/OR tree, and develop algorithms for ranking and unranking the variants of the AND/OR tree.
The developed ranking and unranking algorithms have polynomial time complexity $O(m^2 (n - m))$ for $m < n - 2 m$ and $O(m (n - m)^2)$ for $m > n - 2 m$.
In contrast to the existing algorithms, the new algorithms do not require preprocessing steps and have better time complexity.
\end{abstract}

\section{Introduction}

Combinatorial generation is a branch of science that lies at the intersection of computer science and combinatorics and develops methods and algorithms for ranking and unranking combinatorial sets.
Fundamentals of combinatorial generation can be found in~\cite{Kreher1999}, \cite{Ruskey2003} and \cite{Knuth2011}.
Many~information objects can be represented as a combinatorial set, for example, in the~form of permutations, partitions, trees, paths, and others.
In this case, it is possible to use combinatorial generation methods for processing such sets, their encoding, generation, and storage.
This is especially true when processing large combinatorial sets (when studying different chemical and biological structures).

In this article, we propose to consider the structure of RNA (ribonucleic acid). 
RNA is one of the~main and most important elements of the cellular structure of any organism.
An RNA molecule is a~chain of nucleotides containing four nitrogenous bases: A (adenine), C (cytosine), G (guanine), and U (uracil).
The~sequence of nucleotides in the RNA chain forms the RNA primary structure.
In~mathematical terms, an RNA primary structure is a sequence of elements taken from the set $\{ A, C, G, U \}$, i.e. $a = (a_1, a_2, \ldots )$ where $a_i \in \{ A, C, G, U \}$.

\newpage
In addition, hydrogen bonds can be formed between two bases in the RNA chain (for example, Watson-Crick base-pairs for A--U and C--G).
The formation of hydrogen bonds in an RNA molecule leads to different loops in it, which complicates the~structure of the RNA molecule.
The obtained shape of the~RNA chain with hydrogen bonds forms an RNA secondary structure.
In mathematical terms, an RNA secondary structure can be represented in the form of a graph, where the nodes of the~graph are nucleotides and the edges of the graph show the connections between bases.
For more details on the features of RNA structures, see \cite{Saenger1984}, \cite{Reidys2011} and~\cite{Waterman2018}.

RNA molecules play an important role in cellular nucleic acid processing and gene expression.
According to~\cite{Picardi2015}, a~deeper understanding of these processes can be obtained through the detailed study of RNA structures.
For example, one such research direction is the~prediction of RNA secondary structures: \cite{Zuker1984, Oluoch2019, Sato2021, Sato2022}.
%The practical and theoretical significance of studying RNA secondary structures is shown by a large number of new articles in this field of science.
Moreover, the improvement of such research works through the~development of computational algorithms is one of the main directions of modern bioinformatics.

Let us consider the combinatorial set of RNA primary structures, which will be denoted by $A_n$.
The set of RNA primary structures is described by one parameter $n$ (the number of nucleotides in an RNA chain).
This combinatorial set can be represented as the set of $n$-permutations of $m$ elements with repetitions where $m = 4$.
That is, $n$ elements are selected from the set of $m = 4$ elements $\{A, C, G, U \}$ and these elements can be repeated.
Thus, algorithms for ranking and unranking the set of such permutations can also be used for the combinatorial set of RNA primary structures.

Next, we consider the combinatorial set of RNA secondary structures, which will be denoted by $A_{n,m}$.
This combinatorial set is described by two parameters: $n$ is the number of nucleotides in an RNA chain, $m$ is the number of base-pairs.
There are also researches related to the development of combinatorial generation algorithms for the set of RNA secondary structures: \cite{Nebel2011, Seyedi-Tabari2010, Mohammadi2018}.

\cite{Nebel2011} presented algorithms for generating a set of $m$ random RNA secondary structures of length $n$ (without specifying the number of base-pairs).
These combinatorial generation algorithms were developed using the method of~\cite{Flajolet1994}.
As well as in~\cite{Viennot1985}, an RNA secondary structure was encoded by a correct bracket sequence with the possibility of adding spaces (or~a~Motzkin word).
The developed algorithms have time complexity $O(m n \log{n})$ and require a preprocessing step with time complexity $O(n^2)$.
The main disadvantage of this approach is the~lack of control over the number of base-pairs in the generated RNA secondary structures, as well as the absence of ranking and unranking algorithms.

\cite{Seyedi-Tabari2010} presented algorithms for listing, ranking, and unranking the combinatorial set of RNA secondary structures of length $n$ with $m$ base-pairs.
Then \cite{Mohammadi2018} presented a parallel version of the listing algorithm.
These combinatorial generation algorithms were developed using an approach based on the calculation of sequence prefixes, that is similar to the~method of~\cite{Ryabko1998}.
In this case, an RNA secondary structure was encoded by a tree with $n$ internal nodes and $m$ external nodes, which was also encoded by an $E$-sequence.
The ranking and unranking algorithms have polynomial time complexity $O(n m  - m^2)$ and $O(n^2 + m^2)$, but both of these algorithms require a preprocessing step with time complexity $O((n - m) |A_{n,m}|)$, where $|A_{n,m}|$ is the~cardinality of the set $A_{n,m}$.
Since $|A_{n,m}|$ can be calculated using binomial coefficients (see Equation~\eqref{fRNAnm_binom}), then the preprocessing step has exponential time complexity.
Thus, the preprocessing step is the main disadvantage of this approach, since it has exponential time complexity and requires a large memory space to store the results of preprocessing 
calculations.

In this paper, the main purpose of our research is to develop new effective algorithms (with polynomial time complexity and without preprocessing steps) for ranking and unranking the combinatorial set of RNA secondary structures of length $n$ with $m$ base-pairs.

\section{Main results}

\subsection{Research method}

There are different approaches for developing combinatorial generation algorithms.
In this paper, we propose to apply the~method for developing combinatorial generation algorithms that is based on using AND/OR trees and described in~\cite{Shablya2020}.
To apply this method for developing ranking and unranking algorithms, it is necessary to represent a given combinatorial set in the form of an AND/OR tree structure for which the total number of its variants is equal to the value of the cardinality function of the combinatorial set.

An AND/OR tree is a tree structure that contains nodes of two types: AND nodes and OR nodes (edges to the sons of an~AND node are connected by an additional curve in Figure~\ref{fig:AndOrTree_RNAnm}). %~\cite{Chauvin2004}
A variant of an~AND/OR tree is a~tree structure obtained by removing all edges except one for each OR node.
If we know the~cardinality function $f$ of a combinatorial set $A$ and $f$ belongs to the~algebra $\{ \mathbb{N}, +, \times, R \}$ (usage of only natural numbers, addition and product operations, and the primitive recursion operator), then we can construct the corresponding AND/OR tree structure.
Hence, we need to know the required form of the cardinality function for the considered combinatorial set of RNA secondary structures.

\subsection{Combinatorial set}

For a compact representation, we encode an RNA secondary structure by the corresponding Motzkin word. %~\cite{Viennot1985}
Figure~\ref{fig:RNA_Example} presents an example of an RNA secondary structure of length $n = 27$ with $m = 9$ base-pairs.

\begin{figure}[htbp]
%\begin{center}
	\centering
	\includegraphics[width=14cm]{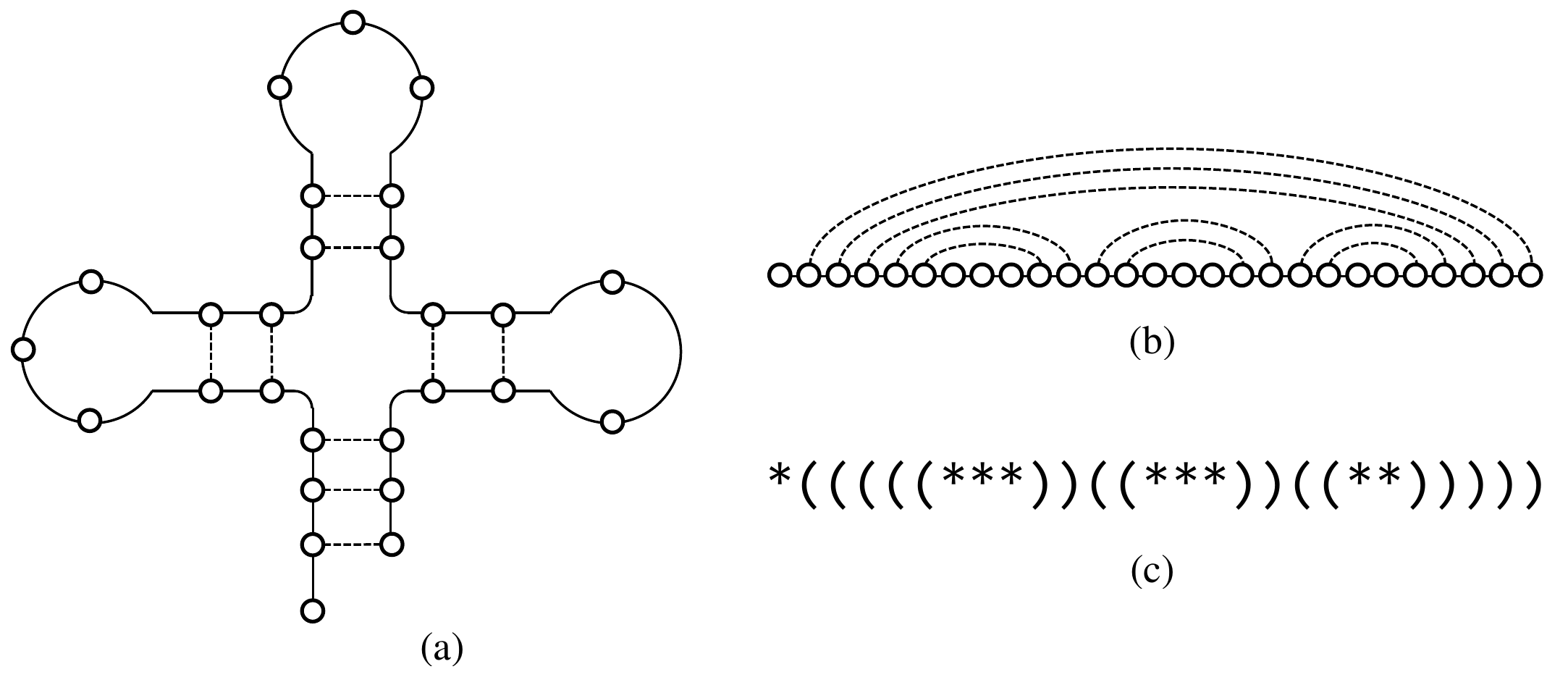}
	\caption{An example of representing an RNA secondary structure of length $n = 27$ with $m = 9$ base-pairs.}
	\label{fig:RNA_Example}
%\end{center}
\end{figure}

Figure~\ref{fig:RNA_Example}(a) shows the "real" shape of the~RNA secondary structure as a chain of nucleotides (nodes connected by solid lines) with base-pairs (nodes connected by dashed lines).
In Figure~\ref{fig:RNA_Example}(b), the RNA secondary structure is presented as the straight line of the RNA chain.
In Figure~\ref{fig:RNA_Example}(c), this RNA secondary structure is represented as the corresponding Motzkin word in the form of a correct bracket sequence with the possibility of adding spaces (denoted by '*').
Hence, we can encode an RNA secondary structure by a sequence $a~=~(a_1, \ldots, a_n)$, where $a_i~\in~\{$~'(',~')',~'*'~$\}$.
In this case, '(' and ')' denote a base-pair, '*'~denotes a nucleotide without hydrogen bonds.

The cardinality function of the considered combinatorial set of RNA secondary structures $A_{n,m}$ is defined by~\cite{Schmitt1994}:%[Equation\ (2)]
\begin{equation}
\label{fRNAnm}
f(n,m) =
|A_{n,m}| =
S_n^m =
S_{n - 1}^m + \sum\limits_{i = 0}^{m - 1}{\sum\limits_{j = 2 i}^{n - 2 (m - i) - 1}{S_{n - 2 - j}^{m - 1 - i} S_j^i}},
\end{equation}
where $S_n^0 = 1$ for $n \geq 0$ and $S_n^m = 0$ for $m \geq \frac{n}{2}$.

The values of $S_n^m$ form the integer sequence $A089732$ in~\cite{OEIS} and can be calculated using the following explicit formula by~\cite{Schmitt1994}:%[Theorem\ 2.2]
\begin{equation}
\label{fRNAnm_binom}
S_n^m = \frac{1}{n - m}\binom{n - m}{m}\binom{n - m}{m + 1}
\end{equation}
for $n > m \geq 0$ and $S_0^0 = 1$.
Table~\ref{tab1} presents the first values of $S_n^m$.

\begin{table}[htbp]
%\begin{center}
\centering
\caption{Several first values of $S_n^m$.}
\label{tab1}
\setlength{\extrarowheight}{2pt}
\begin{tabularx}{0.8\textwidth} { 
  | >{\centering\arraybackslash}X 
  | >{\centering\arraybackslash}X
  | >{\centering\arraybackslash}X
  | >{\centering\arraybackslash}X 
  | >{\centering\arraybackslash}X 
  | >{\centering\arraybackslash}X | }
\hline
&\multicolumn{5}{@{}c|@{}}{$m$}\\
\cline{2-6}
$n$& 0& 1& 2& 3& 4\\
\hline
0&	1&&&&\\
1&	1&&&&\\
2&	1&&&&\\
3&	1&	1&&&\\
4&	1&	3&&&\\
5&	1&	6&	1&&\\
6&	1&	10&	6&&\\
7&	1&	15&	20&		1&\\
8&	1&	21&	50&		10&\\
9&	1&	28&	105&	50&		1\\
10&	1&	36&	196&	175&	15\\
\hline
\end{tabularx}
%\end{center}
\end{table}

Generating functions play an important role in the study of integer sequences. %~\cite{Kucukoglu2018}
The following bivariate generating function of the numbers $S_n^m$ is known from~\cite{OEIS}:
\begin{equation*}
S(x,y) =
\sum\limits_{n \geq 0}{\sum\limits_{m \geq 0}{S_n^m x^n y^m}} =
\frac{1 - x + x^2 y - \sqrt{(1 - x - x^2 y)^2 - 4 x^3 y}}{2 x^2 y}.
\end{equation*}

There is a relationship between the numbers $S_n^m$ and the Narayana numbers $N_n^m$.
It is also worth noting that the Narayana numbers $N_n^m$ numbers are related to the well-known Catalan numbers $C_n$ (for more information see~\cite{Petersen2015}).
In this case, the Narayana numbers $N_n^m$ divide combinatorial sets defined by the Catalan numbers $C_n$ into subsets in accordance with some additional parameter $m$.
For~example, one combinatorial interpretation of the Catalan numbers is the number of Dyck $n$-paths, then the Narayana numbers show the~number of Dyck $n$-paths with $m$ peaks (see~\cite{Kruchinin2022}).

The following bivariate generating function of the Narayana numbers $N_n^m$ is known from~\cite{OEIS}:
%The Narayana numbers $N_n^m$ are defined by the bivariate generating function~\cite{OEIS,Kruchinin2021}
\begin{equation*}
N(x,y) =
\sum\limits_{n > 0}{\sum\limits_{m > 0}{N_n^m x^n y^m}} =
\frac{1 - x - x y - \sqrt{(1 - x - x y)^2 - 4 x^2 y}}{2 x},
\end{equation*}
and for $n \geq m > 0$ we have
\begin{equation*}
N_n^m = \frac{1}{n} \binom{n}{m - 1} \binom{n}{m}.
\end{equation*}

Hence, we can show the relationship of the numbers $S_n^m$ with the Narayana numbers $N_n^m$  through the~relationship of their generating functions by
\begin{equation*}
S(x,y) =
1 + \frac{N(x,x y)}{x y}
\end{equation*}
or their explicit formulas by
\begin{equation*}
S_n^m =
N_{n - m}^{m + 1},
\end{equation*}
where $n > m \geq 0$ and $S_0^0 = 1$.

\subsection{AND/OR tree}

Since Equation~\eqref{fRNAnm} belongs to the algebra $\{ \mathbb{N}, +, \times, R \}$, we can construct the AND/OR tree structure for~$S_n^m$, which is presented in Figure~\ref{fig:AndOrTree_RNAnm}.
In this AND/OR tree:
\begin{itemize}
\item each addition in Equation~\eqref{fRNAnm} is represented as an OR node where all terms are represented as its sons;
\item each multiplication in Equation~\eqref{fRNAnm} is represented as an AND node where all factors are represented as its sons;
\item each recursive operation in Equation~\eqref{fRNAnm} is represented as recursion in the AND/OR tree (denoted by a node with a triangle).
\end{itemize}

\begin{figure}[htbp]
%\begin{center}
	\centering
	\includegraphics[width=13.3cm]{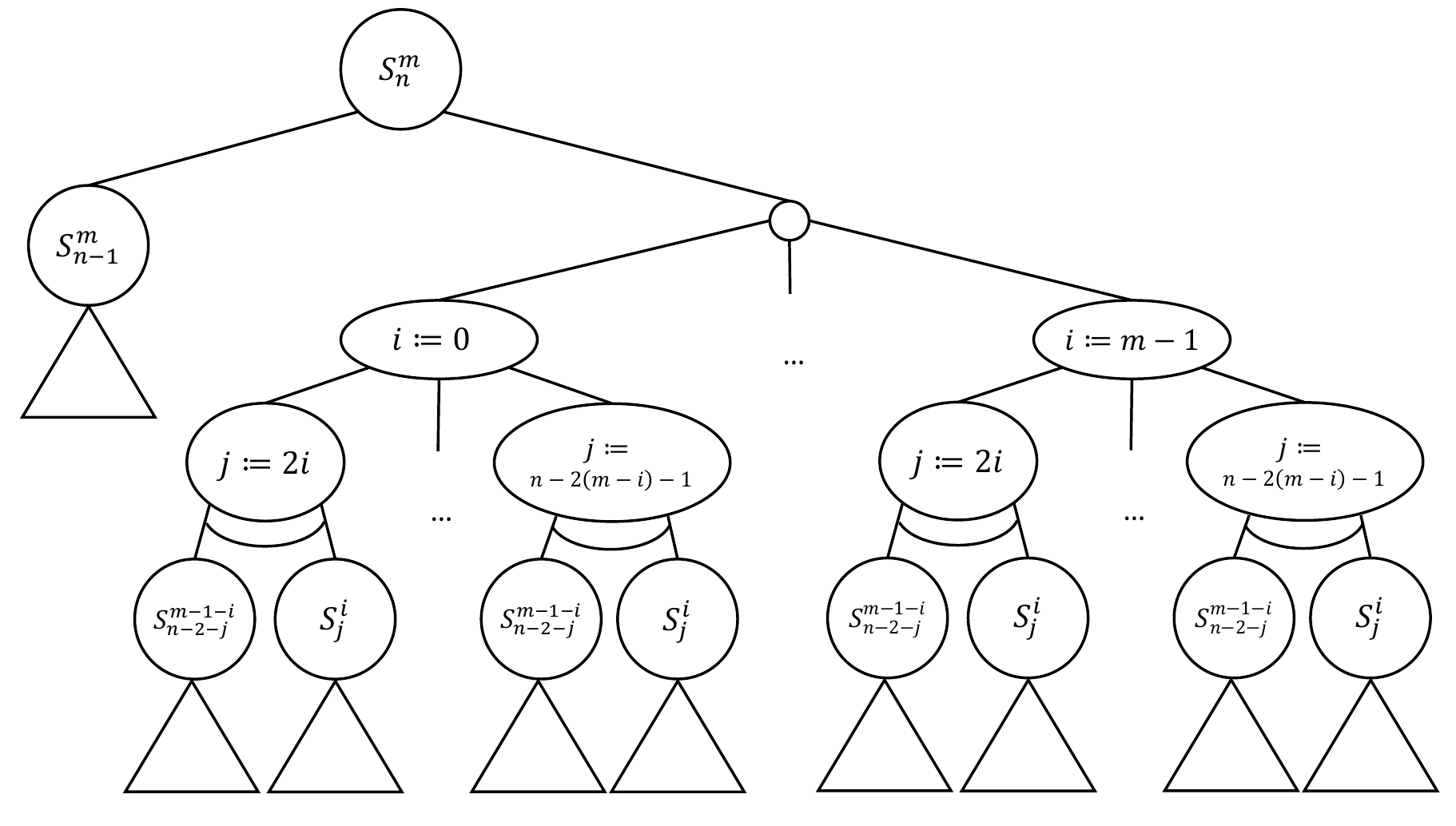}
	\caption{An AND/OR tree for $S_n^m$.}
	\label{fig:AndOrTree_RNAnm}
%\end{center}
\end{figure}

For this AND/OR tree structure, there are the following initial conditions:% that correspond to the initial conditions of Equation~\eqref{fRNAnm}:
\begin{itemize}
\item if we get a node labeled $S_n^0$, then it is a leaf node in the AND/OR tree where $n \geq 0$
\item if we get a node labeled $S_n^m$ where $m \geq \frac{n}{2}$, then we need to remove this node and also remove the~subtree containing this node until the closest son of an OR node.
\end{itemize}

Figure~\ref{fig:AndOrTree_RNAnm_Example} presents an example of an AND/OR tree structure for $S_n^m$ where $n = 6$ and $m = 2$.
The total number of its variants is equal to $S_6^2 = 6$.

\begin{figure}[hp]
%\begin{center}
	\centering
	\includegraphics[width=14cm]{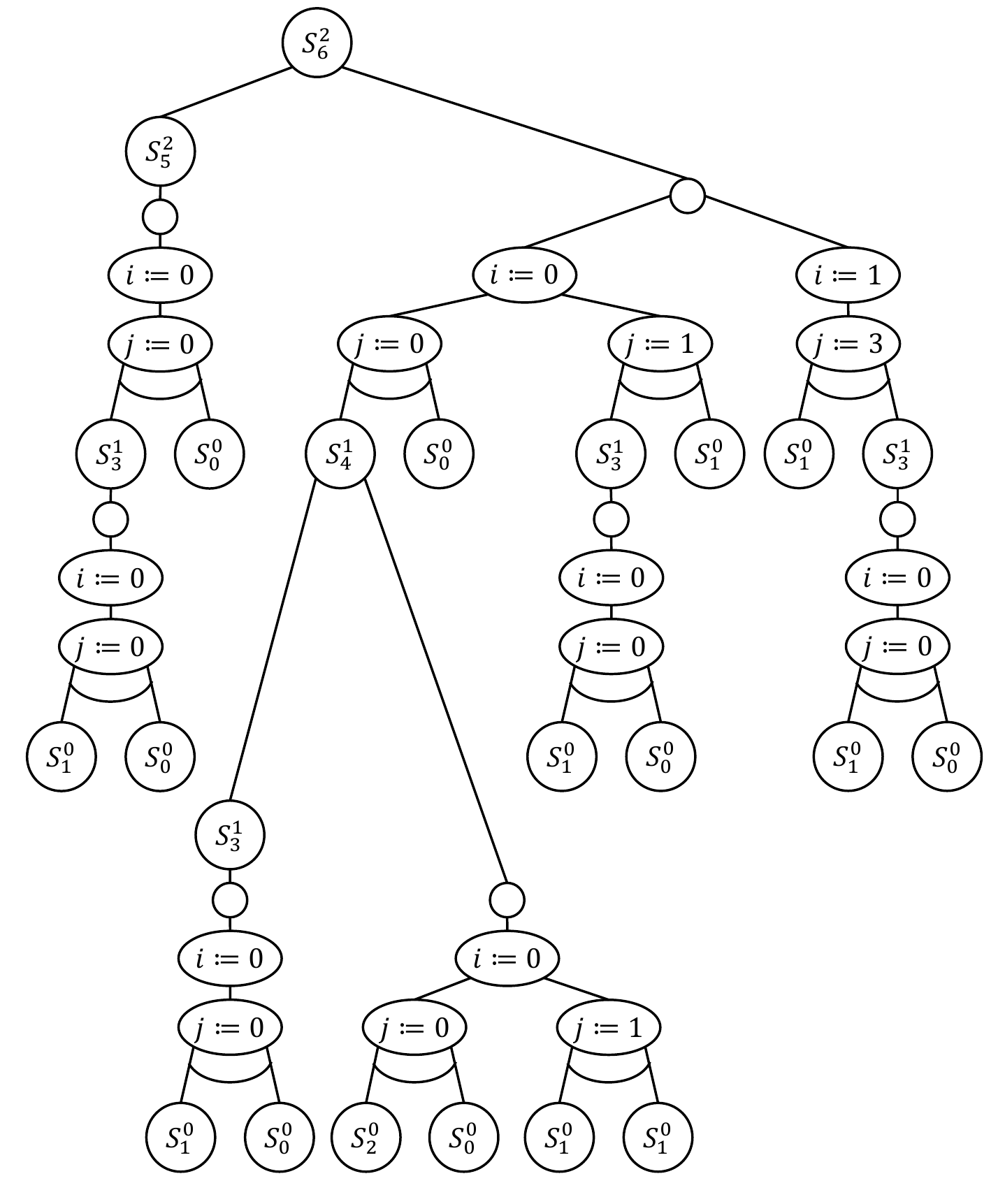}
	\caption{The AND/OR tree for $S_6^2$.}
	\label{fig:AndOrTree_RNAnm_Example}
%\end{center}
\end{figure}

\begin{theorem}
There is a bijection between the set of RNA secondary structures of length $n$ with $m$ base-pairs and the set of variants of the AND/OR tree for $S_n^m$.
\end{theorem}

\begin{proof}
Let represent RNA secondary structures of length $n$ with $m$ base-pairs as correct bracket sequences of length $n$ with $m$ pairs of brackets with the possibility of adding spaces.
The total number of such RNA secondary structures is equal to $S_n^m$.
The~total number of variants of the AND/OR tree for $S_n^m$ presented in Figure~\ref{fig:AndOrTree_RNAnm} is also equal to $S_n^m$.
Hence, a bijection between the set of RNA secondary structures of length $n$ with $m$ base-pairs and the set of variants of the AND/OR tree for $S_n^m$ is defined by the following rules:
\begin{itemize}
\item each selected left son of the OR node labeled $S_n^m$ determines the addition of a space to the sequence $a = (a_1, \ldots, a_n)$: the~sequence $a$ is represented as $a = \texttt{concat}$(('*'),$b$), i.e. the first element of the sequence is the added space ($a_1 = $~'*') and the remaining part of the sequence denoted by $b = (a_2, \ldots, a_n)$ is a correct bracket sequence of length $n - 1$ with $m$~pairs of brackets (it corresponds to the subtree of the node labeled $S_{n - 1}^m$);
\item each selected right son of the OR node labeled $S_n^m$ determines the addition of a pair of brackets to the sequence $a = (a_1, \ldots, a_n)$: the sequence $a$ is represented as $a = \texttt{concat}$(('('),$b$,(')'),$c$), where $b = (a_2, \ldots, a_{n - j - 1})$ is a correct bracket sequence of length $n - 2 - j$ with $m - 1 - i$ pairs of brackets (it corresponds to the subtree of the node labeled $S_{n - 2 - j}^{m - 1 - i}$) and $c = (a_{n - j + 1}, \ldots, a_n)$ is a~correct bracket sequence of length $j$ with $i$ pairs of brackets (it corresponds to the~subtree of the~node labeled $S_j^i$), and the values of $i$ and $j$ are determined by the selected sons in the AND/OR~tree.
\item each leaf node labeled $S_n^0$ determines the addition of $n$ spaces to the sequence $a = (a_1, \ldots, a_n)$.
\end{itemize}
\end{proof}

In the obtained rules for the bijection, the function $\texttt{concat}$ denotes merging sequences, i.e. for two given sequences $a = (a_1, \ldots, a_n)$ and $b = (b_1, \ldots, b_m)$ we get the following sequence:

\begin{equation*}
c = \texttt{concat}(a,b) = (a_1, \ldots, a_n, b_1, \ldots, b_m) = (c_1, \ldots, c_{n + m}).
\end{equation*}

For a compact representation, a variant of the AND/OR tree for $S_n^m$ will be encoded by a sequence $v = (v_1, v_2)$, where:
\begin{itemize}
\item an empty sequence $v = (\,)$ corresponds to the selection of a leaf node labeled $S_n^0$;
\item $v_1$ corresponds to the selection of the left son ($v_1 = 0$) or the right son ($v_1 = 1$) of the OR node labeled $S_n^m$;
\item if $v_1 = 0$, then $v_2$ corresponds to the variant of the subtree of the node labeled $S_{n - 1}^m$;
\item if $v_1 = 1$, then $v_2$ is encoded by a sequence $v_2 = (I, J, vl, vr)$, where:
\begin{itemize}
\item $I$ corresponds to the selected value of $i$ in the AND/OR tree;
\item $J$ corresponds to the selected value of $j$ in the AND/OR tree;
\item $vl$ corresponds to the variant of the subtree of the node labeled $S_{n - 2 - j}^{m - 1 - i}$;
\item $vr$ corresponds to the variant of the subtree of the node labeled $S_j^i$.
\end{itemize}
\end{itemize}

Algorithm~\ref{alg:RNASecondaryStructureToVariant} and Algorithm~\ref{alg:VariantToRNASecondaryStructure} present the developed rules for the bijection between the set of RNA secondary structures and the set of variants of the corresponding AND/OR tree.
In these algorithms, $()$~denotes an empty sequence.

\begin{algorithm2e}[htbp]
\DontPrintSemicolon
\SetKwFunction{ToVariant}{RNASSToVariant}
\caption{An algorithm for converting a correct bracket sequence with the possibility of adding spaces into a variant of the AND/OR tree for $S_n^m$.}
\label{alg:RNASecondaryStructureToVariant}
\ToVariant($a = (a_1, \ldots, a_n)$, $n$, $m$)\\
\Begin
{
	\lIf{$m = 0$}{\Return{$()$}}
	\lIf{$a_1 = $\normalfont{'*'}}{$v := (0, $ \ToVariant($(a_2, \ldots, a_n)$, $n - 1$, $m$)$)$}
	\Else
	{
		$s_1 := 0$\\
		$s_2 := 0$\\
		\For{$i := 1$ \KwTo $n$}
		{
			\lIf{$a_i = $\normalfont{'('}}{$s_1 := s_1 + 1$}
			\lIf{$a_i = $\normalfont{')'}}{$s_2 := s_2 + 1$}
			\If{$s_1 = s_2$}
			{
				$I := m - s_1$\\
				$J := n - i$\\
				$b := (a_2, \ldots, a_{n - J - 1})$\\
				$c := (a_{n - J + 1}, \ldots, a_n)$\\
				\textbf{break}
			}
		}
		$vl := $ \ToVariant($b$, $n - 2 - J$, $m - 1 - I$)\\
		$vr := $ \ToVariant($c$, $J$, $I$)\\
		$v := (1, (I, J, vl, vr))$
	}
	\Return{$v$}
}
\end{algorithm2e}

\begin{algorithm2e}[htbp]
\DontPrintSemicolon
\SetKwFunction{ToObject}{VariantToRNASS}
\SetKwFunction{Concat}{concat}
\caption{An algorithm for converting a variant of the AND/OR tree for $S_n^m$ into a correct bracket sequence with the~possibility of adding spaces.}
\label{alg:VariantToRNASecondaryStructure}
\ToObject($v = (v_1, v_2)$, $n$, $m$)\\
\Begin
{
	\lIf{$m = 0$}{\Return{$(a_1, \ldots, a_n) := ($\normalfont{'*'}$, \ldots, $\normalfont{'*'}$)$}}
	\lIf{$v_1:= 0$}{$a := $ \Concat($($'*'$)$, \ToObject($v_2$, $n - 1$, $m$))}
	\Else{
		$(I, J, vl, vr) := v_2$\\
		$b := $ \ToObject($vl$, $n - 2 - J$, $m - 1 - I$)\\
		$c := $ \ToObject($vr$, $J$, $I$)\\
		$a := $ \Concat($($'('$)$, $b$, $($')'$)$, $c$)
	}
	\Return{$a$}
}
\end{algorithm2e}

Let us consider the computational complexity of these algorithms in more detail:
\begin{itemize}
\item In Algorithm~\ref{alg:RNASecondaryStructureToVariant}, there are $n - 2 m$ recursive calls where $a_1 = $~'*' (each such recursive call requires making a new sequence of length $n - 1$) and there are $m$ recursive calls where $a_1 \neq $~'*' (each such recursive call requires making calculations maximum $n$ times in the for-loop and making new sequences of maximum length $n - 2$).
Hence, assuming algebraic operations with numbers in $O(1)$, Algorithm~\ref{alg:RNASecondaryStructureToVariant} has polynomial time complexity $O(n (n - m) + n m)$;
\item In Algorithm~\ref{alg:VariantToRNASecondaryStructure}, there are maximum $n - 2 m$ recursive calls where $v_1 = 0$ (each such recursive call requires making a~new sequence of length $n$) and there are $m$ recursive calls where $v_1 = 1$ (each such recursive call requires making a new sequence of length $n$).
Hence, assuming algebraic operations with numbers in $O(1)$, Algorithm~\ref{alg:VariantToRNASecondaryStructure} has polynomial time complexity $O(n (n - m) + n m)$.
\end{itemize}

Thus, both presented algorithms have polynomial time complexity $O(n (n - m) + n m)$.
Also, taking into account the constraint $m < \frac{n}{2}$, we get $O(n^2)$.

The obtained bijection make it possible to apply the combinatorial generation method mentioned above to the development algorithms for ranking and unranking the variants of the constructed AND/OR tree structure for $S_n^m$.
Combining these combinatorial generation algorithms with the obtained bijection, we~will get algorithms for ranking and unranking the considered combinatorial set of RNA secondary structures.

\subsection{Ranking and unranking algorithms}

Applying the general approach described in~\cite{Shablya2020}, we develop algorithms for ranking and unranking the variants of the constructed AND/OR tree structure.
Algorithm~\ref{alg:RankVariant_RNASecondaryStructure} presents the required actions for ranking the variants of the AND/OR tree for $S_n^m$, Algorithm~\ref{alg:UnrankVariant_RNASecondaryStructure} presents the required actions for unranking its variants.

Let us consider the computational complexity of these algorithms in more detail:
\begin{itemize}
\item In Algorithm~\ref{alg:RankVariant_RNASecondaryStructure}, there are maximum $n - 2 m$ recursive calls where $v_1 = 0$ (each such recursive call requires making only one assignment) and there are $m$ recursive calls where $v_1 = 1$ (each such recursive call requires making calculations of $S_j^i$ maximum $m$ times for $i$ and $n - 2 m$ times for $j$).
Applying Equation~\eqref{fRNAnm_binom}, the calculation of $S_n^m$ has polynomial time complexity $O(m)$ for $m < n - 2 m$ and $O(n - m)$ for $m > n - 2 m$. 
Hence, assuming algebraic operations with numbers in $O(1)$, Algorithm~\ref{alg:RankVariant_RNASecondaryStructure} has polynomial time complexity $O(m^2 (n - m))$ for $m < n - 2 m$ and $O(m (n - m)^2)$ for $m > n - 2 m$.
\item In Algorithm~\ref{alg:UnrankVariant_RNASecondaryStructure}, there are maximum $n - 2 m$ recursive calls where $r < S_{n - 1}^m$ (each such recursive call requires making only one assignment) and there are $m$ recursive calls where $r \geq S_{n - 1}^m$ (each such recursive call requires making calculations of $S_J^I$ maximum $m$ times for $I$ and $n - 2 m$ times for $J$).
Hence, assuming algebraic operations with numbers in $O(1)$, Algorithm~\ref{alg:UnrankVariant_RNASecondaryStructure} has polynomial time complexity $O(m^2 (n - m))$ for $m < n - 2 m$ and $O(m (n - m)^2)$ for $m > n - 2 m$.
\end{itemize}

Taking into account the constraint $m < \frac{n}{2}$, we get polynomial time complexity $O(n^3)$ for both presented algorithms.
For comparison, the ranking and unranking algorithms of~\cite{Seyedi-Tabari2010} have better polynomial time complexity $O(n^2)$, however they require a preprocessing step that has exponential time and space complexity.
Since the developed algorithms do not require preprocessing steps, they have better overall time and space complexity.
In addition, the developed algorithms can be supplemented with a preprocessing step with the calculation of all required values of $S_n^m$.
As a result, we get the preprocessing step with polynomial time complexity $O(n^2)$ and the new ranking and unranking algorithms with the same polynomial time complexity $O(n^2)$.

\begin{algorithm2e}[htbp]
\DontPrintSemicolon
\SetKwFunction{Rank}{RankVariant\_RNASS}
\caption{An algorithm for ranking a variant of the AND/OR tree for $S_n^m$.}
\label{alg:RankVariant_RNASecondaryStructure}
\Rank($v = (v_1, v_2)$, $n$, $m$)\\
\Begin
{
	\lIf{$m = 0$}{\Return{$0$}}
	\lIf{$v_1 = 0$}{$r := $ \Rank($v_2$, $n - 1$, $m$)}
	\Else
	{
		$(I, J, vl, vr) := v_2$\\
		$l_1 := $ \Rank($vl$, $n - 2 - J$, $m - 1 - I$)\\
		$l_2 := $ \Rank($vr$, $J$, $I$)\\
		$r := l_1 + S_{n - 2 - J}^{m - 1 - I} l_2 + S_{n - 1}^m + \sum\limits_{i = 0}^{I - 1}{\sum\limits_{j = 2 i}^{n - 2 (m - i) - 1}{S_{n - 2 - j}^{m - 1 - i} S_j^i}} + \sum\limits_{j = 2 I}^{J - 1}{S_{n - 2 - j}^{m - 1 - I} S_j^I}$ 
	}
	\Return{$r$}
}
\end{algorithm2e}

\newpage
\begin{algorithm2e}[!htbp]
\DontPrintSemicolon
\SetKwFunction{Unrank}{UnrankVariant\_RNASS}
\caption{An algorithm for unranking a variant of the AND/OR tree for $S_n^m$.}
\label{alg:UnrankVariant_RNASecondaryStructure}
\Unrank($r$, $n$, $m$)\\
\Begin
{
	\lIf{$m = 0$}{\Return{$()$}}
	\lIf{$r < S_{n - 1}^m$}{$v := (0,$ \Unrank($r$, $n - 1$, $m$)$)$}
	\Else
	{
		$r := r - S_{n - 1}^m$\\
		$sum := 0$\\
		$I := 0$\\
		$J := 0$\\
		\While{$sum + S_{n - 2 - J}^{m - 1 - I} S_J^I \leq r$}
		{
			$sum := sum + S_{n - 2 - J}^{m - 1 - I} S_J^I$\\
			\lIf{$J \leq n - 2 (m - I) - 1$}{$J := J + 1$}
			\Else
			{
				$I := I + 1$\\
				$J := 2 I$
			}
		}
		$r := r - sum$\\
		$l_1 := r \mod{S_{n - 2 - J}^{m - 1 - I}}$\\
		$l_2 := \left\lfloor {r}\,/\,{S_{n - 2 - J}^{m - 1 - I}} \right\rfloor$\\
		$vl := $ \Unrank($l_1$, $n - 2 - J$, $m - 1 - I$)\\
		$vr := $ \Unrank($l_2$, $J$, $I$)\\
		$v := (1, (I, J, vl, vr))$
	}
	\Return{$v$}
}
\end{algorithm2e}

Combining Algorithm~\ref{alg:RNASecondaryStructureToVariant} with Algorithm~\ref{alg:RankVariant_RNASecondaryStructure}, we get the algorithm for ranking the combinatorial set of RNA secondary structures of length $n$ with $m$ base-pairs.
The combination of Algorithm~\ref{alg:UnrankVariant_RNASecondaryStructure} with Algorithm~\ref{alg:VariantToRNASecondaryStructure} performs the inverse operation (unranking).
Table~\ref{tab:RNASecondaryStructure1} and Table~\ref{tab:RNASecondaryStructure2} present examples of ranking the combinatorial set of RNA secondary structures of length $n$ with $m$~base-pairs (that are represented as correct bracket sequences of length $n$ with $m$ pairs of brackets with the possibility of adding spaces).

\begin{table}[!htbp]
%\begin{center}
\centering
\caption{An example of ranking the %combinatorial 
set of RNA secondary structures of length $n = 6$ with $m = 2$ base-pairs.}
\label{tab:RNASecondaryStructure1}
\setlength{\extrarowheight}{2pt}
\begin{tabularx}{1.0\textwidth} { 
  | >{\centering\arraybackslash}c 
  | >{\centering\arraybackslash}X
  | >{\centering\arraybackslash}c | }
\hline
Correct bracket sequence & Variant of AND/OR tree & Rank\\
\hline
	$* \, ( \, ( \, * \, ) \, )$ & $(0,(1,(0,0,(1,(0,0,(),())),())))$ & $0$ \\
	$( \, * \, ( \, * \, ) \, )$ & $(1,(0,0,(0,(1,(0,0,(),()))),()))$ & $1$ \\
	$( \, ( \, * \, * \, ) \, )$ & $(1,(0,0,(1,(0,0,(),())),()))$ & $2$ \\
	$( \, ( \, * \, ) \, * \, )$ & $(1,(0,0,(1,(0,1,(),())),()))$ & $3$ \\
	$( \, ( \, * \, ) \, ) \, *$ & $(1,(0,1,(1,(0,0,(),())),()))$ & $4$ \\
	$( \, * \, ) \, ( \, * \, )$ & $(1,(1,3,(),(1,(0,0,(),()))))$ & $5$ \\
\hline
\end{tabularx}
%\end{center}
\end{table}

\newpage
\begin{table}[htbp]
%\begin{center}
\centering
\caption{An example of ranking the %combinatorial 
set of RNA secondary structures of length $n = 8$ with $m = 3$ base-pairs.}
\label{tab:RNASecondaryStructure2}
\setlength{\extrarowheight}{2pt}
\begin{tabularx}{1.0\textwidth} { 
  | >{\centering\arraybackslash}c 
  | >{\centering\arraybackslash}X
  | >{\centering\arraybackslash}c | }
\hline
Correct bracket sequence & Variant of AND/OR tree & Rank\\
\hline
$* \, ( \, ( \, ( \, * \, ) \, ) \, )$ & $(0,(1,(0,0,(1,(0,0,(1,(0,0,(),())),())),())))$ & $0$ \\
	$( \, * \, ( \, ( \, * \, ) \, ) \, )$ & $(1,(0,0,(0,(1,(0,0,(1,(0,0,(),())),()))),()))$ & $1$ \\
	$( \, ( \, * \, ( \, * \, ) \, ) \, )$ & $(1,(0,0,(1,(0,0,(0,(1,(0,0,(),()))),())),()))$ & $2$ \\
	$( \, ( \, ( \, * \, * \, ) \, ) \, )$ & $(1,(0,0,(1,(0,0,(1,(0,0,(),())),())),()))$ & $3$ \\
	$( \, ( \, ( \, * \, ) \, * \, ) \, )$ & $(1,(0,0,(1,(0,0,(1,(0,1,(),())),())),()))$ & $4$ \\
	$( \, ( \, ( \, * \, ) \, ) \, * \, )$ & $(1,(0,0,(1,(0,1,(1,(0,0,(),())),())),()))$ & $5$ \\
	$( \, ( \, * \, ) \, ( \, * \, ) \, )$ & $(1,(0,0,(1,(1,3,(),(1,(0,0,(),())))),()))$ & $6$ \\
	$( \, ( \, ( \, * \, ) \, ) \, ) \, *$ & $(1,(0,1,(1,(0,0,(1,(0,0,(),())),())),()))$ & $7$ \\
	$( \, ( \, * \, ) \, ) \, ( \, * \, )$ & $(1,(1,3,(1,(0,0,(),())),(1,(0,0,(),()))))$ & $8$ \\
	$( \, * \, ) \, ( \, ( \, * \, ) \, )$ & $(1,(2,5,(),(1,(0,0,(1,(0,0,(),())),()))))$ & $9$ \\
\hline
\end{tabularx}
%\end{center}
\end{table}

\section{Conclusions}

In this paper, we study the combinatorial set of RNA secondary structures of length $n$ with $m$ base-pairs.
For a compact representation, an RNA secondary structure is encoded by a correct bracket sequence with the possibility of adding spaces.
For this combinatorial set, we construct an AND/OR tree structure, find a bijection between the combinatorial set and the set of variants of the AND/OR tree, and develop algorithms for ranking and unranking the variants of the AND/OR tree.

The developed ranking and unranking algorithms have polynomial time complexity $O(m^2 (n - m))$ for $m < n - 2 m$ and $O(m (n - m)^2)$ for $m > n - 2 m$.
The conducted computational experiments also confirmed the derived time complexity of the~algorithms.

The existing algorithms of Seyedi-Tabari et al. for ranking and unranking the considered combinatorial set of RNA secondary structures of length $n$ with $m$ base-pairs also have polynomial time complexity.
However, the operation of these ranking and unranking algorithms requires a preprocessing step that has exponential time and space complexity.
In contrast to the algorithms of Seyedi-Tabari et al., the new algorithms do not require preprocessing steps.
Thus, the developed ranking and unranking algorithms have better time complexity (the resulting complexity is polynomial).

%\acknowledgements
%At the end of the manuscript, right before the bibliography you might want to place an acknowledgment. This can be easily done by using the command \verb!\acknowledgements! as you can see here.

\nocite{*}
\bibliographystyle{abbrvnat}
% use the following instead if you encounter problems 
%\bibliographystyle{alpha}
\bibliography{article}%sample-dmtcs}
\label{sec:biblio}

\end{document}